\documentclass[11pt,letterpaper]{article}

\usepackage[left=1in, right=1in, top=1in, bottom=1in]{geometry}
\usepackage{amssymb,amsthm,amsmath,mathtools}
\usepackage{braket}
\usepackage{xcolor}
\usepackage{algorithm}
\usepackage{graphicx}
\usepackage{abstract}
\usepackage{setspace}
\usepackage[noend]{algpseudocode}
\usepackage{authblk}

\usepackage[hidelinks, final=true,colorlinks = false,allcolors = {blue}]{hyperref}
\usepackage[capitalize]{cleveref}

\newtheorem{theorem}{Theorem}
\newtheorem{proposition}{Proposition}
\newtheorem{lemma}{Lemma}
\newtheorem{corollary}{Corollary}
\newtheorem{definition}{Definition}

\newtheorem{claim}{Claim}

\newtheorem{remark}{Remark}

\DeclareMathOperator{\im}{\mathrm{im}}

\newcommand{\eps}{\varepsilon}
\newcommand{\rk}{\mathrm{rk}}
\newcommand{\full}{\mathrm{full}}

\newcommand{\daniel}[1]{\textcolor{violet}{Daniel: #1}}
\newcommand{\simon}[1]{\textcolor{blue}{Simon: #1}}

\title{Holey graphs: very large Betti numbers are testable}

\author[1]{D\'aniel Szab\'o}
\author[1]{Simon Apers}
\affil[1]{Universit\'e Paris Cit\'e, CNRS, IRIF, Paris, France\vspace{10pt}}
\date{}

\begin{document}
\begin{titlepage}
\maketitle
\thispagestyle{empty}

\begin{abstract}
\normalsize

We show that the graph property of having a (very) large $k$-th Betti number $\beta_k$ (over $\mathbb{Z}_2$) for constant $k$ is testable with a constant number of queries in the dense graph model.
More specifically, we consider a clique complex defined by an underlying graph and prove that for any $\eps>0$, there exists $\delta(\eps,k)>0$ such that testing whether $\beta_k \geq (1-\delta) d_k$ for $\delta \leq \delta(\eps,k)$ reduces to tolerantly testing $(k+2)$-clique-freeness, which is known to be testable.
% More specifically, we can decide in constant time whether $\beta_k$ is large, or if the graph is far from that.
This complements a result by Elek (2010) showing that Betti numbers are testable in the bounded-degree model.
%Our result combines matroid theory and the graph removal lemma.
For our result we consider simplicial complexes as combinatorial objects, and combine matroid theory and the graph removal lemma.
\end{abstract}

\vspace*{\fill}

\small This preprint has not undergone peer review (when applicable) or any post-submission improvements or corrections. The Version of Record of this contribution is published in SOFSEM 2025: Theory and Practice of Computer Science, and is available online at \url{https://doi.org/10.1007/978-3-031-82697-9_22}.

\end{titlepage}

\section{Introduction}

There has been an increasing interest in the notion of ``topological data analysis'' (TDA), where data is interpreted as a topological object called a ``simplicial complex'', and the homology of the complex is used to robustly classify the data.
An important example of such a complex is the Vietoris-Rips or flag complex -- this complex is obtained by associating a graph to the input data (e.g., with edges indicating similarity between data points), and the complex consisting of all subsets that induce cliques in the graph.
% For example a set of data points can be represented as a topological object based on the pairwise distances between the points (Vietoris-Rips complex).
An important feature of this simplicial complex is the so-called Betti number, or intuitively the number of high dimensional ``holes''.
In particular, the so-called persistent Betti numbers have been useful for applications because they capture a scale-independent global property of the data set \cite{pranav_cosmic,krishnapriyan2021machine,bukkuri_oncology}.
Unfortunately, calculating Betti numbers exactly is computationally hard \cite{crichigno2022clique}, and so the task of estimating them has been in the focus of interest \cite{chazal15,lloyd2016quantum,hayakawa2021quantum,berry2022quantifying,apers2023simple}.
In this work we use the lense of ``(graph) property testing'' to further our understanding of Betti number estimation.
% As our main result, we show that the property of having a (very) large Betti numbers can be easily recognized.
% examine the related setting of property testing from a combinatorial point of view.

In graph property testing we wish to decide whether a graph has a certain property, or whether it is ``far'' from having that property \cite{goldreich2010introduction}.
In the \emph{dense} graph model, an $n$-vertex graph is $\eps$-far from having a property if we have to add or remove more than $\eps n^2$ edges for the graph to have the property.
A \emph{tester} for a given property is a randomized algorithm that, given query access to the adjacency matrix of a graph $G$, can distinguish with constant success probability whether $G$ has that property or is $\eps$-far from having it.
A graph property is said to be \emph{testable} if there exists a tester that makes a number of queries that is a function only of $\eps$, and so independent of the graph size.
Examples of testable properties are bipartiteness, triangle-freeness and, more generally, monotone (closed under removing edges) and hereditary (closed under removing vertices) graph properties \cite{alon2005every,alon2008characterization}.

In this work we prove the following theorem (for a more formal statement see \Cref{thm:testing_Betti_formal}).

\begin{theorem}[Informal] \label{thm:Betti-testable}
The property of having a (very) large $k$-th Betti number is testable for constant $k$.
\end{theorem}
\noindent
Strictly speaking, we consider the $k$-th Betti number of the \emph{clique complex} associated to $G$.
This is the simplicial complex defined by the family of all vertex subsets that induce a clique in $G$.
On an intuitive level, the $k$-th Betti number $\beta_k$ of this complex counts the number of independent $k$-dimensional ``holes'' in the complex, which is bounded by the number of $k$-cliques $d_k$ in $G$.
More formally, $\beta_k$ equals the rank of the $k$-th \emph{homology group}.
We prove \cref{thm:Betti-testable} by showing that, for any constant $k$ and $\eps>0$, there exists $\delta(\eps,k)>0$ such that testing whether $\beta_k \geq (1-\delta) d_k$ for some fixed $\delta \leq \delta(\eps,k)$ reduces to tolerantly testing $(k+2)$-clique-freeness.
We prove the result for
\[
\delta(\eps,0) = \sqrt{2\eps} \; , \qquad
\delta(\eps,1) = \eps/3 \; , \qquad
\delta(\eps,k) = 1/\mathrm{tower}(k^4\log(1/\eps)) \;\; (k>1).
\]
Here the $\mathrm{tower}(\ell)$-function denotes a height-$\ell$ tower of powers of 2's -- this explains the extra quantifier in ``\emph{(very)} large Betti numbers''.
Nonetheless, this property is neither trivial for constant $k$ and $\eps$ nor monotone or hereditary.
To see this, consider the $(k+1)$-partite graph which has $d_k = (\frac{n}{k+1})^{k+1}$ and $\beta_k = (\frac{n}{k+1}-1)^{k+1}$ \cite[Proposition 1]{berry2022quantifying}, and so $\beta_k/d_k = 1 - O(k^2/n)$.
This shows that for any $k$ and $\delta>0$ there exist graphs with the property $\beta_k \geq (1-\delta) d_k$. (Compare this to our \Cref{claim:construction} which states that any large clique complex that has few $k$-faces is close to having a large $k$-th Betti number.)
Moreover, the quantity $\beta_k/d_k$ increases as a function of $n$, so the property cannot be monotone or hereditary.

% (the complete $(k+1)$-partite graph has $\beta_k/d_k = 1 - O(k^2/n)$ \cite{adamaszek2014extremal}) nor monotone (both the empty and the complete graph have $\beta_k = 0$ for $k\ge1$) or hereditary (for the $\beta_k/d_k$).

\paragraph{Betti numbers in the bounded-degree model.}
This work was partly motivated by the result of Elek \cite{elek2010betti} in the \emph{bounded-degree model}.
In this model, the graph is assumed to have a constant bound $d$ on the vertex degrees, and a query reveals the (at most $d$) neighbours of a vertex.
Elek showed that, for any $\eps > 0$ and with a number of queries only dependent on $\eps$, it is possible to return an estimate $\hat\beta_k$ satisfying $\hat\beta_k = \beta_k \pm \eps n$ for $k < d$ (for $k \geq d$ necessarily $d_k = \beta_k = 0$).
The proof is based on (sparse) graph limits, and is very different from our approach.

Unfortunately, such a result in the dense graph model is not possible: returning an estimate $\hat\beta_k = \beta_k \pm \eps n$ (or even $\hat\beta_k = \beta_k \pm \eps d_k$) requires $\Omega(n)$ many queries in the dense graph model.
To see this, consider the case $k=0$ for which $d_0 = n$ and $\beta_0$ equals the number of connected components of the graph.
The cycle graph has $\beta_0 = 1$ while any graph with $\leq n/2$ edges has $\beta_0 \geq n/2$.
However, it takes $\Omega(n)$ queries to distinguish these graphs in the dense model.
This motivates the weaker formulation of large Betti number testability that we use in our work.\footnote{Note also that the contrapositive, having a small Betti number, is trivial to test.
E.g., a graph cannot be far from having small Betti number $\beta_0$ since we can always add a cycle, thereby setting $\beta_0 = 1$.}
% \simon{reference later construction: ``any sparse graph is close to graph with large normalized Betti number ($\beta_k/d_k \approx 1$).''}

\paragraph{Quantum algorithms for large Betti numbers.}
Another motivation came from a recent stream of works on quantum algorithms for estimating Betti numbers (see e.g.~\cite{lloyd2016quantum,gyurik2022towardsquantum,hayakawa2021quantum,crichigno2022clique}).
Under certain conditions (e.g., a form of well-conditionedness), these works suggest an exponential quantum speedup over classical algorithms, returning an estimate $\hat\beta_k = \beta_k \pm \eps n^k$ in time $\mathrm{poly}(n,k,1/\eps)$.
Such an estimate is only relevant for very large Betti numbers ($\beta_k$ scaling with $n^k$), and a stringent question is whether ``typical'' graphs can have such large Betti numbers.
In a recent work \cite{apers2023simple} %, coauthored by the authors of this work, 
a classical benchmark algorithm was proposed based on path-integral Monte Carlo.
The algorithm has a polynomial runtime in certain regimes, narrowing down the conditions for a potential exponential quantum speedup.

The current work investigates this question from a new (property testing) perspective, and it yields a tool to investigate whether typical graphs can have a (very) large Betti number.

\paragraph{Open questions.}
Our work raises a number of open questions.
The most obvious one is whether our results can be pushed further.
E.g., it might be possible to test more moderately sized Betti numbers, or Betti numbers for non-constant $k$ (the case of interest for quantum algorithms).

Another open direction is to generalize the framework of graph property testing to simplicial complexes more generally.
By limiting ourselves to clique complexes, we could phrase our results in the graph property testing language, but this might not be the most natural approach.

\paragraph{Outline of proof and paper.}
In \Cref{sec:prelim} we formally introduce property testing, simplicial complexes, and the necessary matroid theory.

The first contribution appears in \Cref{sec:independence}.
We use the matroid notion of independence to relate the Betti number $\beta_k$ to the number of ``independent'' $K_{k+2}$ cliques in the graph, and we bound the total number of cliques as a function of the number of independent cliques.
%Then we show that having a large Betti number is equivalent to having few independent $K_{k+2}$ subgraphs.% (sic).

Finally, in \Cref{sec:testing-betti}, we build on these tools to reduce the problem of testing large Betti numbers to that of (tolerantly) testing clique-freeness, which is known to be testable.
In particular, we show that having a large Betti number implies that the graph is close to being $K_{k+2}$-free, while being far from having a large Betti number implies that the graph is far from being $K_{k+2}$-free.

\section{Preliminaries} \label{sec:prelim}

In this section we introduce necessary but well-known preliminaries on graph property testing, simplicial complexes and matroid theory.

\subsection{Property testing and subgraph freeness}
In graph property testing we want to decide if an $n$-vertex input graph $G$ has a property $P$ or if it is $\eps$-far from any graph satisfying $P$. In the dense graph model we have query access to elements of the adjacency matrix of $G$, and we want to minimise the number of queries we make. This model was introduced in \cite{GGR98}.

\begin{definition}[$\eps$-far] \label{def:epsfar}
    The distance of two graphs $G$ and $G'$ is defined as
    $$D(G,G')=\frac{\min_\pi\{G\triangle\pi(G')\}}{n^2},$$
    where $\pi$ is any permutation of the vertices and $\triangle$ denotes the symmetric difference of the two edge sets.
    We say that $G$ is \emph{$\eps$-far} from property $P$ if $D(G,G')>\eps$ for all $G'$ satisfying $P$.
\end{definition}

\begin{definition}[Property tester]
    In the dense graph model, a randomised algorithm $A$ is a property testing algorithm for property $P$ if given query access to the adjacency matrix of the input graph $G$, it satisfies the following.
    \begin{itemize}
        \item If $G$ satisfies $P$ then $A$ returns ``YES'' with probability $\ge 2/3$.
        \item If $G$ is $\eps$-far from $P$ then $A$ returns ``NO'' with probability $\ge 2/3$.
    \end{itemize}
    A \emph{tolerant} property testing algorithm satisfies a stronger constraint: for some $\eps_1<\eps_2$, it distinguishes (i) $G$ being $\eps_1$-close to $P$, from (ii) $G$ being $\eps_2$-far from $P$.
\end{definition}

We call a property $P$ \textit{(tolerantly) testable} if there is a (tolerant) property testing algorithm such that the number of queries it makes is independent of the input length.

The following well-known lemma (that can be proved using the Szemerédi regularity lemma \cite{szemeredi}) has been central to proving many testability results, and we will also use it.
\begin{lemma}[Graph removal lemma, \cite{Furedi95}] \label{lemma:graph_removal}
    For any graph $H$ and any $\eps>0$ there exists a $\delta>0$ such that the following holds: any $n$-vertex graph $G$ ($|V(H)|<n$) that contains at most $\delta n^{|V(H)|}$ copies of $H$ as subgraphs can be made $H$-free by removing at most $\eps n^2$ edges (i.e. $G$ is $\eps$-close to being $H$-free).
\end{lemma}

It follows almost directly from this result that, for any constant-sized graph $H$, the property of being $H$-free is testable~\cite{alon1994algorithmic}.
Combined with the fact that every testable property in the dense graph model is also \emph{tolerantly} testable \cite{fischer07}\footnote{More precisely, they prove that for every testable property there is a distance approximation algorithm. This implies tolerant testability.}, we get the following lemma which we are going to use later.

\begin{lemma} \label{lemma:tolerant_testing}
    For any graph $H$ the property of $H$-freeness is tolerantly testable in the dense graph model. The number of queries depends only on the distance parameters $\varepsilon_1,\varepsilon_2$ and on $|H|$.
\end{lemma}

\subsection{Simplicial complexes}
A simplicial complex is a downward closed set family over a set $V$ of vertices.
As such, it can be thought of as a higher-dimensional generalisation of graphs (albeit more restrictive than hypergraphs).
\begin{definition}[Simplicial complex]
    An (abstract) simplicial complex $\Delta$ is a set of subsets of the vertex set $V$, such that if $S\in \Delta$ and $S'\subset S$ then also $S'\in\Delta$.
    The sets in $\Delta$ with cardinality $k+1$ are called the $k$-faces of $\Delta$.
\end{definition}

We are going to denote the set of $k$-faces of a complex by $F_k(\Delta)=\{S\in\Delta, |S|=k+1\}$ and its size by $d_k(\Delta)=|F_k(\Delta)|$. When it is clear from the context which simplicial complex is being considered, we will write only $F_k$ and $d_k$. If the largest subset of $V$ that is in the complex is of cardinality $D+1$ then we say that the complex is $D$-dimensional.

A \textit{clique complex} is a special case of a simplicial complex, and is defined by some underlying graph $G$. The sets in the clique complex associated to $G$ are exactly the cliques of $G$. This implies for instance that a size-$(k+1)$ subset $S\subseteq V$ is in the complex if (and only if) all the size-$k$ subsets of $S$ are in the complex.

% For our analysis we wish to consider unoriented faces.
% This is in contrast to most of the works on simplicial homology, but for our combinatorial arguments it is sufficient and makes more intuitive sense.
% \simon{discuss this sentence}

The $k$-chain group $C_k$ of a simplicial complex $\Delta$ over an Abelian group $\textbf{\textrm{G}}$ is defined as $C_k=\{\sum_{i=1}^{d_k}\alpha_i S_i\}$ where $S_i\in F_k(\Delta)$ and $\alpha_i\in \textbf{\textrm{G}}$.
Many sources consider integer coefficients ($\textbf{\textrm{G}}=\mathbb{Z}$), but for our purpose it suffices to pick binary coefficients $\textbf{\textrm{G}} = \mathbb{Z}_2$.
This lies closest to our combinatorial interpretation of homology, which is based on \emph{unoriented} faces and for which the chain group $C_k=2^{F_k}$ is then simply the set of all subsets of $F_k$.
% unoriented case we are going to consider only binary coefficients ($\textbf{\textrm{G}}=\mathbb{Z}_2$).
% This corresponds to taking $C_k=2^{F_k}$, the set of all the subsets of $F_k$.
As a consequence, we will refer to the elements of $C_k$ either as a sum of $k$-faces or as a set of $k$-faces -- the two are equivalent.
%, where the orientation of a face corresponds to an ordering of the vertices in that face. For this, it is enough to assume that the vertices are labelled by integers from 1 to $n$, and each face of the complex has an orientation corresponding to the increasing order of its vertices.

For each $k>0$, the $k$-th boundary operator $\delta_k$ is a homomorphism that maps a $k$-face to the sum of the $(k-1)$-faces that ``surround'' the $k$-face.
\begin{definition}[Boundary operator]
    For any $k\ge 1$ the $k$-th boundary operator is a homomorphism $\delta_k:C_k\to C_{k-1}$.
    For $S\in F_k(\Delta)$ and $S=\{v_1,v_2,\dots, v_{k+1}\}$ it is defined by $\delta_k(S)=\sum_{i=1}^{k+1}S\setminus\{v_i\}$.
\end{definition}
\noindent
It follows from this definition that the boundary of a boundary is always zero, i.e.,~$\delta_{k}(\delta_{k+1}(.))=0$.
Sometimes we are going to use the same notation for \emph{boundary vector}s: for $S\in F_k$, $\delta_k(S)\in\{0,1\}^{d_{k-1}}$, where a coordinate is 1 iff the corresponding $(k-1)$-face appears in the boundary of $S$.

Now we can define the Betti numbers, that are at the center of interest in this article.
\begin{definition}[Betti number] \label{def:Betti}
    The $k$-th Betti number $\beta_k$ of a simplicial complex $\Delta$ is the rank of the $k$-th homology group:
    $$\beta_k(\Delta)=\textnormal{rk}(\ker(\delta_{k})/\im(\delta_{k+1})).$$
\end{definition}
More intuitively, define a $k$-dimensional hole as a subset $H \subseteq F_k(\Delta)$ that \emph{has} no boundary and \emph{is} no boundary.
Equivalently,  $\delta_k(H)=0$ (and so $H \in \ker(\delta_k)$) and there exists no $H' \subseteq F_{k+1}(\Delta)$ such that $\delta_{k+1}(H') = H$ (and so $H \notin \im(\delta_{k+1})$).
Then $\beta_k$ counts the number of ``independent'' $k$-dimensional holes in the complex -- where we formalize the notion of independence in the next section.

% \begin{remark}
%     Taking integer coefficients in $C_k$ would mean that we are interested in whether the boundary of a weighted sum of $k$-faces is zero. But for us each face can appear at most once in $C_k$. This would still let us consider oriented faces by taking $\textbf{\textrm{G}}=\{-1,0,+1\}$,but it does not make much difference. If an element of $C_k$ has 0 boundary over $\{-1,0,+1\}$ then its boundary is also 0 over $\mathbb{Z}_2$, and if an element of $C_k$ has 0 boundary over $\mathbb{Z}_2$ then there is an orientation such that its boundary is also 0 over $\{-1,0,+1\}$.
% \simon{discuss this}
% \end{remark}

% Our result strongly relies on the so-called Euler characteristic.
% This is a topological invariant that connects the number of faces in a complex to its Betti numbers.
% \begin{definition}[Euler characteristic] \label{def:euler_char}
%     The Euler characteristic of a simplicial complex $\Delta$ is
%     $$\chi(\Delta)=\sum_{k=0}^{\infty} (-1)^k d_k(\Delta)
%     = \sum_{k=0}^{\infty} (-1)^k \beta_k(\Delta).$$
% \end{definition}

\subsection{Matroids}
The appropriate notion of independence of faces and of holes comes from matroid theory. A matroid is a downward closed set family with an additional property\footnote{In this sense matroids are a specialisation of simplicial complexes, although we are going to use them in a different way.} called the exchange property.
\begin{definition}[Matroid]
    A matroid $M$ over ground set $E$ is a family of subsets $I\subseteq 2^E$ called the independent subsets of $E$, and which satisfies the following properties.
    \begin{enumerate}
        \item $\emptyset\in I$.
        \item If $A\in I$ and $B\subseteq A$ then $B\in I$.
        \item If $A,B\in I$ and $|B|<|A|$ then $\exists v\in A\setminus B$ such that $B\cup\{v\}\in I$.
    \end{enumerate}
\end{definition}

The easiest example of a matroid is a graph. In this case, the ground set $E$ in the matroid is the edge set of the graph, and we call a subset of edges independent if it is cycle-free. Matroids that can be defined this way by a graph are called \textit{graphic matroid}s or cycle matroids.

Another important example is linear independence of vectors. The elements of $E$ are vectors from a vector space, and a subset of them is called independent if the vectors are linearly independent (over a field $F$). Matroids that can be defined in this way are called \textit{linear matroid}s (or representable over $F$).

The \emph{simplicial matroid} (or simplicial geometry) $M_k(\Delta)$ associated to a simplicial complex $\Delta$ is a linear matroid defined as follows.
It appears in e.g.~\cite{CrapoRota, simpMatroids}.
\begin{definition}[Simplicial matroid] \label{def:simpl-matroid}
The $k$-simplicial matroid $M_k(\Delta)$ associated to a simplicial complex $\Delta$ is the linear matroid whose ground set is the set of boundary vectors $\delta_k(S) \in \{0,1\}^{d_{k-1}}$ for $S \in F_k(\Delta)$.
\end{definition}
\noindent
Motivated by this, we call a subset of $k$-faces independent if the corresponding boundary vectors are linearly independent (over the field $\{0,1\}$).

A maximal independent set of a matroid $M$ is called a \textit{basis}. It is well known that all the bases of a matroid have the same size, equal to the \textit{rank} $\mathrm{rk}(M)$ of the matroid.
The full $k$-simplicial matroid $M_k(\Delta_k^{\full})$ is the $k$-simplicial matroid associated to the full complex $\Delta_k^{\full}=\{S\subseteq V, |S|\le k+1\}$ that contains all the $k+1$-subsets as $k$-faces, but it does not have any higher dimensional face.
%The following result will be useful about the rank of the full $k$-simplicial matroid $M_k(\Delta_k^{\full})$, associated to the full complex $\Delta_k^{\full}$ that contains all the $k+1$-subsets as $k$-faces.

\begin{proposition} [e.g. \cite{simpMatroids}, Proposition 6.1.5] \label{claim:rank_full}
    $\mathrm{rk}(M_k(\Delta_k^{\full})) = \binom{n-1}{k}$.
\end{proposition}

For a construction, fix a vertex $u$ of $\Delta_k^{\full}$ and take the set of $k$-faces that contain $u$. It is easy to see that this set of size $\binom{n-1}{k}$ is a basis of the matroid.

\section{Betti numbers via independent faces} \label{sec:independence}

In this section we connect the number of independent $k$-faces with the total number of $k$-faces, and connect the Betti number $\beta_k$ to the number of independent $k$- and $(k+1)$-faces in the complex.
%While we could find in the literature variants of some of the claims in this section (see e.g.~\Cref{remark:literature}), the key claims were unknown to us before.

%\subsection{Some facts about our notion of independence} \label{subsec:indep}
%We use the concept of matroids to define independence of sets of $k$-faces.
% Consider a simplicial complex $\Delta$ and an integer $k>0$. To each $k$-face $S$ of $\Delta$ we associate its \textit{boundary vector}, i.e. a vector over $\{0,1\}^{d_{k-1}}$ that corresponds to the boundary $\delta_k(S)$. We will take the linear matroid (over field $\{0,1\}$) defined by this set of $d_k$ many vectors.
% \begin{definition}
%     In a simplicial complex $\Delta$ we call a set of $k$-faces independent if the corresponding set of boundary vectors is independent in the linear matroid defined above.
% \end{definition}

The notion of independence of faces in \Cref{def:simpl-matroid} leads to the following useful observation.
It is a direct consequence of the fact that a set of $k$-faces has zero boundary if and only if the sum of the corresponding boundary vectors is the zero vector.

\begin{proposition} \label{claim:independence}
In a $k$-dimensional simplicial complex (i.e., $|F_{k+1}|=0$), a set of $k$-faces is independent iff no subset of them forms a $k$-dimensional hole.
\end{proposition}
% \begin{proof}
% First, let us assume that there is a set $A$ of $k$-faces that form a hole. By definition, the boundary of the hole is zero. Using \Cref{obs:boundary}, this means that the sum of the boundary vectors in the hole is zero, thus the boundary vectors corresponding to the $k$-faces in $A$ are linearly dependent (over field $\{0,1\}$).

% If, on the contrary, a set $A$ of $k$-faces forms no hole then by the fact that there are no $(k+1)$-faces in the complex (thus no set of $k$-faces can be in $\im(\delta_{k+1})$), we know that no subset of $A$ has zero boundary, otherwise that would be a hole by definition. Using \Cref{obs:boundary} again, this means that there is no subset of $A$ where the sum of the boundary vectors is zero, thus the boundary vectors corresponding to the $k$-faces in $A$ are linearly independent.
% \end{proof}

The independence of holes is defined similarly.
A $k$-dimensional hole is a set of $k$-faces (an element of $C_k$), and associated to it is a characteristic vector over $\{0,1\}^{d_{k}}$. We associated the same kind of (boundary) vectors to $(k+1)$-faces: in this sense a $k$-dimensional hole can be seen as the boundary of a virtual $(k+1)$-dimensional object. 
This way, a set of holes is independent if the corresponding vectors are linearly independent (over field $\{0,1\}$).
An analogue of \Cref{claim:independence} tells us that a set of $k$-dimensional holes is independent iff no subset of them (as virtual $(k+1)$-faces) forms a $(k+1)$-dimensional hole.

Let us denote the rank $\rk(M_k(\Delta))$ of the $k$-simplicial matroid (i.e. the size of a maximal independent set of $k$-faces in $\Delta$) by~$r_k(\Delta)$. We are going to need a lower bound on this value in terms of the total number of $k$-faces $d_k(\Delta)$.
For the sake of completeness, we also include an upper bound in the statement.

\begin{lemma} \label{claim:faces}
For any $0\le k<n$ and any simplicial complex $\Delta$ 
    $$\frac{k+1}{n}d_k(\Delta) \le r_k(\Delta) \le \min\left\{d_k(\Delta), \binom{n-1}{k}\right\}.$$
\end{lemma}
\begin{proof}
    Trivially, $r_k(\Delta) \le d_k(\Delta)$.
    Moreover, the set of $k$-faces $F_k(\Delta)$ of any complex $\Delta$ can be obtained from that of the full complex $F_k(\Delta^{\full}_k)$ by removing faces, and this can only decrease the rank.
    Combined with \Cref{claim:rank_full} we hence get $r_k(\Delta) \le \binom{n-1}{k}$.
    
    % From \Cref{claim:rank_full} we know that the rank of the full $k$-simplicial matroid is $\binom{n-1}{k}$. (For a construction, fix a vertex $u$ of $\Delta_k^{\full}$ and take the set of $k$-faces that contain $u$. It is easy to see that this set of size $\binom{n-1}{k}$ is a basis of the matroid.)
    % Now we can obtain any $k$-dimensional complex $\Delta$ on $n$ vertices by removing some faces of $\Delta_k^{\full}$. So in $\Delta$ if we take a maximal independent set of $k$-faces of size $r_k(\Delta)$, this set is also going to be independent in $\Delta_k^{\full}$, hence, $r_k(\Delta)\le \binom{n-1}{k}$.

    Now let us prove the main part of the claim, which is the lower bound.
    We use a similar argument to the one below \Cref{claim:rank_full}.
    Let $u$ be a vertex in $\Delta$ that is included in a maximum number of $k$-faces (i.e., the vertex with the highest ``$k$-face-degree'').
    These $k$-faces that contain $u$ are independent because each contains a $(k-1)$-face that the others do not (the one without $u$), and this is a non-zero element in their boundary vector. As there are $d_k$ many $k$-faces in $\Delta$, each incident to $k+1$ vertices, the average ``$k$-face-degree'' of a vertex is $(k+1)d_k/n$. Thus the independent set of $k$-faces defined by $u$ has at least this many $k$-faces, and so $r_k\ge (k+1)d_k/n$.
\end{proof}

The next lemma shows a nice connection between the rank, the number of faces and the Betti number.
For  $k=0$ the formula gives the well-known graph formula $c = n - t$, where $t$ is the number of edges in a spanning forest, $n$ is the number of vertices, and $c$ is the number of connected components.
When $k=1$ and the underlying graph is connected and planar, it gives the Euler formula $n+f=e+2$ (with $n$ the number of vertices, $f$ the number of faces surrounded by edges and $e$ the number of edges) because $\beta_1=f-1$, $d_1=e$, $r_1=n-1$ and $r_2=0$. 

%We give two proofs for the lemma, the first one is new and combinatorial. The second one is a simpler, algebraic proof that was pointed out to us by an anonymous reviewer and later we also found it in \cite{TDAlecturenotes} Proposition 3.13.

\begin{lemma}[\cite{TDAlecturenotes} Proposition 3.13.] \label{lemma:rank_faces_holes}
    For any simplicial complex, $\beta_k=d_k-r_k-r_{k+1}$.
\end{lemma}
\begin{proof}
    By applying the rank–nullity theorem to $\delta_k$, we get $d_{k}=\dim \ker \delta_{k}+\dim \im \delta_{k}$. Now notice that $\dim \im \delta_k=r_k$ because the independence of $k$-faces is defined through their boundary vectors (\Cref{def:simpl-matroid}), thus we have $\dim \ker \delta_{k}=d_k-r_k$. From \Cref{def:Betti} we can see that $\beta_k=\dim \ker \delta_k - \dim \im \delta_{k+1}$. Substituting what we got before we obtain $\beta_k=(d_k-r_k)-r_{k+1}$.
\end{proof}

For the interested reader we also we also give an alternative, combinatorial proof of this lemma in \Cref{appendix:proof}, which ties closer to the spirit of this work.

In the special case where $k=0$ and the graph defined by the vertices and edges of $\Delta$ is connected, we have $\beta_0=1$, $d_0=n$ and $r_{k+1}=n-1$ (a spanning tree of the graph is a maximal independent edge set). Thus, $r_0$ has to be defined as 0, which makes sense if we think about $r_k$ as $r_k=\dim(\im(\delta_k))$.

\begin{remark} \label{remark:literature}
Let $T_\ell = \{S \subseteq V, \, |S| = \ell+1\}$ and $A\subseteq T_{k}$. In early works \cite{CrapoRota, Cordovil} only complexes of the form $\Delta^{\full}_{k-1} \cup A$ are analysed in detail. This family of complexes is not enough to express the $k$-th Betti number of an arbitrary simplicial complex. For example, for this restricted class of complexes Cordovil \cite[Proposition 1.2]{Cordovil} showed that $r_k=d_k-\beta_k$, which is only a special case of \Cref{lemma:rank_faces_holes} (with $r_{k+1}=0$).
%Our truncate-and-fill complex is of the form $\Delta^{\full}_{k-1} \cup A \cup B$ for $A \subseteq T_k$ and $B \subseteq T_{k+1}$.
\end{remark}

%In the general case, when there can be $(k+1)$-faces in the complex, instead of the above claim we get $r_k(\Delta)\le d_k(\Delta)-\beta_k(\Delta)$. Indeed, adding $(k+1)$-faces leaves $r_k$ and $d_k$ unchanged, yet it decreases $\beta_k$ (a $(k+1)$-face fills a $k$-dimensional hole).
An easy consequence of \cref{lemma:rank_faces_holes} is the following statement.
\begin{proposition} \label{claim:betti_upperbound}
    For any simplicial complex $\Delta$ and $k\ge 1$, $\beta_k(\Delta)\le \binom{n-1}{k+1}$.
\end{proposition}
\begin{proof}
    In $\Delta^{\full}_k$ we have $\beta_k=d_k-r_k-0=\binom{n}{k+1}-\binom{n-1}{k}=\binom{n-1}{k+1}$, and removing $k$-faces or adding $(k+1)$-faces cannot increase this value.
\end{proof}

\section{Testing large Betti numbers} \label{sec:testing-betti}

In this section we turn to our main result, proving that we can test whether a Betti number is large. We now state our main theorem again.

\begin{theorem}[formal version of \Cref{thm:Betti-testable}] \label{thm:testing_Betti_formal}
    Consider a clique complex $\Delta$.
    For any constant $k$ and $\eps>0$, there exists $\delta(\eps,k)>0$ such that the property of having $k$-th Betti number $\beta_k(\Delta) \geq (1-\delta) d_k$ (over $\mathbb{Z}_2$) is testable for any $\delta \leq \delta(\eps,k)$ (with distance parameter $\eps$).
\end{theorem}

%\begin{remark}
Even though our results in \Cref{sec:independence} hold for general simplicial complexes, the main theorem is restricted to clique complexes. The reason for this is that we wish to phrase our results in the well-established setting of graph property testing. By constraining ourselves to clique complexes, having a large Betti number becomes a graph property (of the underlying graph) rather than a property of an abstract simplicial complex. Also, this way we can use some previous results from graph property testing, like the tolerant testability of subgraph freeness (\Cref{lemma:tolerant_testing}).
%\end{remark}

\subsection{Warm-up: testing many components} \label{sec:beta-0}

The $0$-th Betti number $\beta_0$ of a clique complex $\Delta$ equals the number of connected components of the underlying graph $G$.
As an informal warm-up and a blueprint for the general case, we show how to test whether $\beta_0 \geq (1-\delta) n$.
The argument involves two reductions.

First, we argue that having a large 0-th Betti number is equivalent to having few independent edges. 
%For this, using \Cref{claim:filledBetti}, we can consider the truncate-and-fill complex $\Delta_0$ that only includes the vertices and edges of the graph $G$.
From \Cref{lemma:rank_faces_holes} we get that $\beta_0 = n - r_1-r_0$ where $r_0=0$.
Thus, $\beta_0 \geq (1-\delta) n$ is equivalent to $r_1\le \delta n$, and so testing large $\beta_0$ reduces to testing whether $G$ has a small number of independent edges.

Now comes the second reduction, in which we argue that testing whether $G$ has few independent edges can be reduced to \emph{tolerantly} testing edge-freeness.
For this, note that if $G$ has $r_1 \leq \delta n$ independent edges then the total number of edges $|E| \leq \binom{\delta n+1}{2} < \delta^2 n^2/2 + O(n)$,\footnote{From \Cref{claim:faces} we would get $|E|\le r_1 n/2 \le \delta n^2/2$. We get the better bound by noticing that if there are $\delta n$ independent edges, then we have a maximum number of edges if all the independent edges are in the same connected component and this component is a $K_{\delta n+1}$.} and so $G$ must be $1.1 \delta^2/2$-close to being edge-free.
On the other hand, if $G$ is $\eps$-far from having $r_1 \leq \delta n$, then $G$ must also be $\eps$-far from having $r_1=0$, i.e from being edge-free.
So we reduced the problem of testing $\beta_0 \geq (1-\delta) n$ to that of tolerantly testing edge-freeness (with parameters $\eps_1 = 1.1\delta^2/2$ and $\eps_2 = \eps$).
It remains to note that edge-freeness is tolerantly testable by \Cref{lemma:tolerant_testing}.
%(as it is a monotone graph property) and hence tolerantly testable (as all testable properties in the dense graph model are tolerantly testable).

\subsection{General case} \label{sec:beta-k}

We now turn to proving our general result (\Cref{thm:testing_Betti_formal}), that having a Betti number $\beta_k \geq (1-\delta) d_k$ is testable for constant $k$.
Following the blueprint from the previous section, we first reduce the problem to testing whether there are few independent $(k+1)$-faces, and then reduce testing few independent $(k+1)$-faces to tolerantly testing $(k+2)$-clique freeness.

We consider a clique complex $\Delta$ with underlying graph $G$. 
% Following \Cref{claim:filledBetti}, we will actually consider the truncate-and-fill complex $\Delta_{k}$.
% Applying \Cref{claim:truncatedEuler} to $\Delta_k$ yields
% $$\beta_k
% = d_k-(d_{k+1}-\beta_{k+1}(\Delta_{k}))-\left(\binom{n-1}{k}-\beta_{k-1}(\Delta_{k})\right),$$
% where we use that $\beta_k$, $d_k$ and $d_{k+1}$ are the same in $\Delta$ as in $\Delta_{k}$, which is why we do not write explicitly the complex they correspond to.
% Now, by \Cref{claim:rank_faces_holes}, $d_{k+1}-\beta_{k+1}(\Delta_{k}) = r_{k+1}$ equals the number of independent $(k+1)$-faces.
% Moreover, by \Cref{claim:betti_upperbound} the last term $c_k \coloneqq \left(\binom{n-1}{k}-\beta_{k-1}(\Delta_{k})\right)$ is non-negative and $c_k \in O(n^k)$.
From \Cref{lemma:rank_faces_holes} we get that
\begin{equation} \label{eq:beta_k}
d_k - r_{k+1} \geq \beta_k = d_k - r_{k+1} - r_k,
\end{equation}
from which we can prove the following lemma.

% For the first reduction we rely on the Euler characteristic $\chi$ of $\Delta_{k+1}$ 
% %consisting of all $d_\ell$ $(\ell+1)$-cliques in $G$ for $\ell \leq k+1$.
% We can equivalently express $\chi(\Delta_{k+1})$ as
% \[
% \chi(\Delta_{k+1})
% = \sum_{\ell = 0}^{k+1} (-1)^\ell d_\ell(\Delta_{k+1})
% = \sum_{\ell = 0}^{k+1} (-1)^\ell \beta_\ell(\Delta_{k+1}),
% \]
% where we recall that $\beta_\ell(\Delta_{k+1}) = \beta_\ell(\Delta)$ for $\ell \leq k$ and $\beta_{k+1}(\Delta_{k+1})$ counts the number of independent $(k+1)$-dimensional holes in $\Delta_{k+1}$.
% This allows us to rewrite
% \[
% \beta_k
% = d_k - (d_{k+1} - \beta_{k+1}(\Delta_{k+1})) + (-1)^k \sum_{\ell=0}^{k-1} (-1)^\ell (d_\ell - \beta_\ell),
% \]
% where, by \Cref{claim:rank_faces_holes}, $r_{k+1} = d_{k+1} - \beta_{k+1}(\Delta_{k+1})$ counts the number of independent $(k+1)$-faces in $\Delta_{k+1}$ (or the number of independent $K_{k+2}$'s in the underlying graph $G$).
% Using that\footnote{For the first inequality, we use that by the Euler characteristic $\chi(\Delta_{k})$ we have $d_k - \beta_k(\Delta_{k}) + (-1)^k \sum_{\ell=0}^{k-1} (-1)^\ell (d_\ell - \beta_\ell) = 0$. Since $d_k - \beta_k(\Delta_{k}) \geq 0$ this implies the inequality.} \[
% 0
% \leq -(-1)^k \sum_{\ell=0}^{k-1} (-1)^\ell (d_\ell - \beta_\ell) \in O(n^k),
% \]
% we get
% \begin{equation} \label{eq:beta_k}
% d_k - r_{k+1}
% \geq \beta_k
% \geq d_k - r_{k+1} - O(n^k).
% \end{equation}

% From this we directly get the following lemma.
\begin{lemma}[Large Betti number $\preceq$ few independent cliques] \label{lem:betti-to-independent}
In a clique complex $\Delta$ with underlying graph $G$, if $\beta_k \geq (1-\delta) d_k$ then $r_{k+1} \leq \delta d_k$.
If $G$ is $\eps$-far from having $\beta_k \geq (1-\delta) d_k$ in $\Delta$ then it is $\eps/2$-far from having $r_{k+1}=0$, i.e., $G$ is $\eps/2$-far from $K_{k+2}$-freeness.%\footnotemark{}
% Assume\footnote{We can test this by testing $K_{k+1}$-freeness. If it is not satisfied, we can add a $k$-partite graph on $d_k^{1/(k+1)} \in o(n)$ vertices to get $\beta_k/d_k \approx 1$.} that $d_k \geq n^{k+1-.01}$.
% Then testing whether $G$ has $\beta_k \geq (1-\delta) d_k$ or is $\eps$-far reduces to testing whether $G$ has at most $\delta d_k$ independent $(k+2)$-cliques or is $1.1\eps$-far.
\end{lemma}

%\footnotetext{For this last point, note that by \Cref{lemma:rank_faces_holes} we have that $r_{k+1} \leq \delta d_k - r_k$ is equivalent to $\beta_k \geq (1-\delta) d_k$.
%Hence, a graph that is $\eps$-far from $\beta_k \geq (1-\delta) d_k$ must also be $\eps$-far from $r_{k+1} \leq \delta d_k - r_k$ (and so certainly from being $K_{k+2}$-free).}

\begin{proof}
The first part of the claim is clear from \cref{eq:beta_k}. For the second part (being $\eps$-far) we will use the definition of $\eps$-far (\Cref{def:epsfar}). I.e. we want to prove that if every graph that is $\eps$-close to $G$ has $\beta_k<(1-\delta)d_k$ then every graph that is $\eps/2$-close to $G$ satisfies $r_{k+1}>0$.

For contradiction assume that there is a particular $H$ that is $\eps/2$-close to $G$ but has $r_{k+1}=0$. Since $H$ is $\eps$-close to $G$, it has $\beta_k<(1-\delta)d_k$, or equivalently $r_k+r_{k+1}>\delta d_k$ (using \Cref{lemma:rank_faces_holes}). Because of this, we have $d_k<r_k/\delta\le \binom{n-1}{k}/\delta$ (by \Cref{claim:faces}). With the construction of \Cref{claim:construction} below, we can modify $H$ to get an $H'$ that is $\alpha=\eps/2$-close to $H$ (thus still $\eps$-close to $G$) and that has $\beta_k\ge(1 - \delta)d_k$. This contradicts the assumption that every graph that is $\eps$-close to $G$ satisfies $\beta_k<(1 - \delta)d_k$.
\end{proof}

\begin{proposition} \label{claim:construction}
    Consider a graph $H=(V,E)$ with $|V|=n$ sufficiently large, and assume that $H$ has at most $\binom{n-1}{k}/\delta$ $k$-faces.
    Then for any constant proximity parameter $\alpha$, there is another graph $H'$ that is $\alpha$-close to $H$ and has $\beta_k\ge(1 - \delta)d_k$.
\end{proposition}
\begin{proof}
We give a construction that modifies $H$ to get $H'$. Let us choose any vertex set $S\subseteq V$ of size $|S|=\alpha n$.
We delete all the edges that go between $S$ and $V\setminus S$, and we modify the edges within $S$ to construct a complete $(k+1)$-partite subgraph.
This modifies at most $\alpha n^2$ edges, and so yields a graph $H'$ that is $\alpha$-close to $H$.

The subgraph of $H'$ induced by $S$ is a complete $(k+1)$-partite graph with $\left(\frac{\alpha n}{k+1}\right)^{k+1}$ many $k$-faces.
Thus, in $H'$ we have at most this amount plus the number of original $k$-faces of $H$, i.e. $d_k(H')\le \left(\frac{\alpha n}{k+1}\right)^{k+1} + \binom{n-1}{k}/\delta$.
The number of independent $k$-holes in the subgraph of $H'$ induced by $S$ is $\left(\frac{\alpha n}{k+1}-1\right)^{k+1}$ (see for instance \cite[Proposition 1]{berry2022quantifying}), so in $H'$ it is at least this much: $\beta_k(H')\ge \left(\frac{\alpha n}{k+1}-1\right)^{k+1}$.
Clearly $\beta_k(H')/d_k(H') = 1 - O(1/n)$.
For any $\delta > 0$ this is at least $1-\delta$ for $n$ sufficiently large.
% As a consequence, in $H'$ we have $\beta_k\ge(1-\delta)d_k$ if $\delta\left(\frac{\alpha n}{k+1}\right)^{k+1}\ge \binom{n-1}{k}(\frac{1}{\delta}-1)+\sum_{i=1}^{k+1}
% (-1)^{i-1}\binom{k+1}{i}\left(\frac{\alpha n}{k+1}\right)^{k+1-i}$. The left hand side is $\Omega(n^{k+1})$ and the right hand side is $O(n^k)$, so for large $n$ this holds for any~$\delta$.
\end{proof}

\begin{remark}
    In \Cref{lem:betti-to-independent} the second proximity parameter is not necessarily half of $\eps$, it can be arbitrarily close to it. E.g. it could be $0.99\eps$, but then we have to use the construction of \Cref{claim:construction} with $\alpha=0.01\eps$ instead of $\eps/2$.
\end{remark}

For our second reduction, we use \Cref{claim:faces}, which tells us that if $r_{k+1} \leq \delta d_k$ then 
\[
d_{k+1}
\leq \frac{\delta}{k+2} n d_k
\leq \frac{\delta}{k+2} n \binom{n}{k+1}
\leq \frac{\delta}{(k+2)!} n^{k+2}.
\]
Combined with \Cref{lem:betti-to-independent} we get that $\beta_k \geq (1-\delta) d_k$ implies $d_{k+1}\leq \frac{\delta}{(k+2)!} n^{k+2}$.
We see that a large Betti number implies a small number of $K_{k+2}$'s in the graph, while being far from having a large Betti number implies being far from $K_{k+2}$-freeness (by \Cref{lem:betti-to-independent}).

In fact, by the graph removal lemma (\Cref{lemma:graph_removal}), a small number of $K_{k+2}$'s implies that the graph is close to being $K_{k+2}$-free.
More specifically,
%if $\delta$ is sufficiently small then such a graph is $\eps'=\eps'(\delta,k)$-close to being $K_{k+2}$-free.
%
%\begin{lemma}[Graph removal lemma, \cite{Furedi95,fox2011new}] \label{lemma:graph_removal}
%    For any $k$-vertex graph $H$ and any $\eps>0$ there exists a $\delta = \delta(k,\eps) >0$ such that the following holds: any $n$-vertex graph $G$ ($k<n$) that contains at most $\delta n^k$ copies of $H$ can be made $H$-free by removing at most $\eps n^2$ edges.
%\end{lemma}
%
%As a consequence, 
for any $\eps'>0$ there exists $\delta = \delta(k,\eps') > 0$ such that if $G$ has at most $\frac{\delta}{(k+2)!} n^{k+2}$ many $K_{k+2}$'s then $G$ is $\eps'$-close to being $K_{k+2}$-free.
By picking (say) $\eps' = \eps/2$, it follows that we can test whether $\beta_k \geq (1-\delta) d_k$ by tolerantly testing whether $G$ is $\eps/2$-close or $\eps$-far from $K_{k+2}$-freeness.
By \Cref{lemma:tolerant_testing} we know that $K_{k+2}$-freeness is indeed tolerantly testable, and this proves our main \Cref{thm:testing_Betti_formal}.

To finish, we comment on the scaling of $\delta(k,\eps)$ (the complexity of the algorithm is dominated by $1/\delta$).
The current best upper bound in the graph removal lemma requires $\delta(k,\eps) \leq 1/\mathrm{tower}(5 (k+2)^4 \log(1/\eps))$ \cite{fox2011new}, where $\mathrm{tower}(i)$ is a tower of twos of height $i$ (e.g., $\mathrm{tower}(3) = 2^{2^2}$).
For the case of $k=0$ we could avoid this: recall from \cref{sec:beta-0} that $r_1 \leq \delta n$ implies that $G$ is $\delta^2/2$-close to being edge-free.
Similarly, for $k=1$ we can get a better bound: $r_2 \leq \delta n^2$ implies that $G$ is $3 \delta$-close to being triangle-free. Indeed, if we remove all the edges of a maximal independent triangle set (at most $3\delta n^2$ edges), then any remaining triangle in the graph would contradict the maximality of the chosen set.
We leave the extension of similar arguments to higher $k$ for future work.
In conclusion, we get a tester that distinguishes $\beta_k \geq (1-\delta) d_k$ from being $\eps$-far under the constraints
\[
\delta < \sqrt{2\eps} \;\; (k=0), \qquad
\delta < \eps/3 \;\; (k=1), \qquad
\delta < 1/\mathrm{tower}(5(k+2)^4 \log(1/\eps)) \;\; (k>1).
\]

\section{Acknowledgements}
The authors were partially supported by
French projects EPIQ (ANR-22-PETQ-0007), QUDATA (ANR-18-CE47-0010) and QUOPS (ANR22-CE47-0003-01), and EU project QOPT (QuantERA ERA-NET Cofund 2022-25).

We thank anonymous reviewers for suggestions on improving the manuscript (and for suggesting the algebraic proof of \Cref{lemma:rank_faces_holes}).

\bibliographystyle{alpha}
\bibliography{biblio}

\appendix
\section{An alternative, combinatorial proof of \Cref{lemma:rank_faces_holes}} \label{appendix:proof}

\begin{proof}
    The proof goes by induction. Let $\Delta$ denote the simplicial complex being considered and let us take a basis of the $k$-simplicial matroid over $\Delta$.
    For the base case, we consider the subcomplex where this is the set of all $k$-faces and all the higher dimensional faces are removed, in which case $r_k=d_k$ and $\beta_k=r_{k+1}=0$ and so the formula holds.
    In the inductive step we will put back all of the removed faces.
    We start by adding the rest of the $k$-faces one by one, and we argue that each added face creates exactly one new independent hole.
    %Let us assume for contradiction that either there is a step when no new independent hole is created, or there is a step where more than one new independent hole is created.
    %If adding an edge does not create any new independent hole
    
    First, note that adding a dependent $k$-face $S$ to the complex creates at least one hole (otherwise we could have added it to the basis by \Cref{claim:independence}).
    Moreover, the hole is independent of the previous ones because it contains the face $S$, which no other hole contains so far.
    %So $r_k(\Delta)\ge d_k(\Delta)-\beta_k(\Delta)$.

    Then, we prove that adding a $k$-face creates at most one hole.
    By contradiction, assume that there is a $k$-face $S$ such that when added to the set, more than one new independent holes are created. We consider two of them, $\{S,R_1,\dots,R_p\}$ and $\{S,T_1,\dots,T_q\}$, which we call the ``$R$-hole'' and the ``$T$-hole''. Necessarily they have zero boundary (we denote the boundary vectors the same way as the $k$-faces):
    \begin{align*}
        S+R_1+\dots+R_p&=0 \\
    S+T_1+\dots+T_q&=0.
    \end{align*}
    Adding the equations shows that $\{R_1,\dots,R_p,T_1,\dots,T_q\}$ must also be a hole, call it the ``$RT$-hole''.
    It does not contain $S$, and so must have been present before adding $S$.
    However, by construction, the $R$-, $T$- and $RT$-holes are not independent, and so we get a contradiction.

    Let $\Delta_k$ denote the complex we have now: it contains exactly the faces of $\Delta$ up to dimension $k$, and no faces of higher dimension. So far we proved that $r_k=d_k-\beta_k(\Delta_k)$.
    Let us consider the set of ``potential $k$-holes'' in $\Delta_k$, i.e. sets $H$ of cardinality $k+2$ where all the $(k+1)$-subsets of $H$ are in $\Delta_k$. These are those holes of $\Delta_k$ that may be filled by $(k+1)$-faces in $\Delta$.

    Now we continue the induction by adding to $\Delta_k$ the $(k+1)$-faces of $\Delta$ one by one (and in the end the higher dimensional faces as well) to get back $\Delta$.
    Each $(k+1)$-face fills a potential hole, and it is independent of the previously added ones if and only if the hole being filled is independent of the previously filled ones (as they are the same subset). Thus, every time $\Delta$ gains an independent $(k+1)$-face it loses an independent $k$-hole. This finishes the proof, as adding faces of dimension larger than $k+1$ does not change any parameter in the claim.
\end{proof}

\end{document}